\documentclass[onecolumn,journal,11pt]{IEEEtran}
\makeatletter
\let\ps@IEEEtitlepagestyle\ps@plain
\makeatother

\usepackage[T1]{fontenc}
\usepackage{url}              
\usepackage{cite}             

\usepackage[cmex10]{amsmath}  
\usepackage{hyperref}
\usepackage{mleftright}       
\mleftright                   

\usepackage{graphicx}         
\usepackage{pgfplots}
\pgfplotsset{compat=newest}
\usepgfplotslibrary{external}
     \usetikzlibrary{plotmarks}                         
\usepackage{booktabs} 
\usepackage[symbol]{footmisc}

\usepackage{amsmath,amsthm,amssymb}
\usepackage{enumerate}
\usepackage{color}
\usepackage{xcolor}
\usepackage{url}
\usepackage{balance}
\usepackage{graphicx} 
\usepackage{mathrsfs}
\usepackage{epsfig}
\usepackage{verbatim}
\usepackage{setspace}
\usepackage{cite}
\usepackage{bbm}

  \usepackage{pgfplots}
  \pgfplotsset{compat=newest}
  \usetikzlibrary{plotmarks}
  \usetikzlibrary{arrows.meta}
  \usepgfplotslibrary{patchplots}
  \usepackage{grffile}

\usepackage{multicol}
\usepackage{lipsum}
\usepackage{etoolbox}
\usepackage{algpseudocode}
\usepackage{algorithmicx}
\usepackage{algorithm}
\DeclareMathOperator*{\argmin}{arg\,min}
\DeclareMathOperator*{\argmax}{arg\,max}

\usepackage{pgfplots}
  \pgfplotsset{compat=newest}
  \usetikzlibrary{plotmarks}
  \usetikzlibrary{arrows.meta}
  \usepgfplotslibrary{patchplots}
  \usepackage{grffile}
  \usepackage{amsmath}

\newtheorem{theorem}{Theorem}

\newtheorem{corollary}{Corollary}
\newtheorem{definition}{Definition}

\newtheorem{remark}{Remark}
\newtheorem{example}{Example}

\newcommand{\bbE}{\mathbb{E}}

\newcommand{\cX}{{\cal X}}

\newcommand{\removed}[1]{}

\makeatletter
\patchcmd{\@IEEEeqnarray}{\relax}{\relax\intertext@}{}{}
\makeatother

\algdef{SE}[DOWHILE]{Do}{doWhile}{\algorithmicdo}[1]{\algorithmicwhile\ #1}%

\usepackage[utf8]{inputenc} 
\usepackage[T1]{fontenc}
\usepackage{url}
\usepackage{ifthen}
\usepackage{cite}

\usepackage{tikz}
\usetikzlibrary{positioning, shapes.geometric}
\usepackage{pgfplots}
\pgfplotsset{plot coordinates/math parser=false}
\pgfplotsset{compat=newest}
\usepackage{caption}
\usepackage{subcaption}

\usepackage[cmex10]{amsmath}

\def\BibTeX{{\rm B\kern-.05em{\sc i\kern-.025em b}\kern-.08em
    T\kern-.1667em\lower.7ex\hbox{E}\kern-.125emX}}
\begin{document}

\title{Log-Likelihood Loss for Semantic Compression\\

}

\author{
\IEEEauthorblockN{
{{Anuj Kumar Yadav\IEEEauthorrefmark{1}\IEEEauthorrefmark{2},
Dan Song\IEEEauthorrefmark{2},
Yanina Shkel\IEEEauthorrefmark{1},
Ayfer Özgür\IEEEauthorrefmark{2}
}}}
\vspace*{4mm}\\
\IEEEauthorblockA{\IEEEauthorrefmark{1}School of Computer \& Communication Sciences, EPFL, Switzerland}\\
\IEEEauthorblockA{\IEEEauthorrefmark{2}Department of Electrical Engineering, Stanford University, USA\\
Emails: \{anuj.yadav, yanina.shkel\}@epfl.ch\\ \hspace{13mm}\{songdan, aozgur\}@stanford.edu}
}

\maketitle
\pagestyle{plain}
\begin{abstract}
We study lossy source coding under a distortion measure defined by the negative log-likelihood induced by a prescribed conditional distribution $P_{X|U}$. This \emph{log-likelihood distortion} models compression settings in which the reconstruction is a semantic representation from which the source can be probabilistically generated, rather than a pointwise approximation. We formulate the corresponding rate–distortion problem and characterize fundamental properties of the resulting rate–distortion function, including its connections to lossy compression under log-loss, classical rate–distortion problems with arbitrary distortion measures, and rate–distortion with perfect perception.
\end{abstract}
\smallbreak

\section{Introduction}
Given a source \( X \sim P_X \) taking values in an alphabet \( \mathcal{X} \), and a  conditional distribution \( P_{X|U} \) such that \( P_{X|U}(\cdot | u) \) is a valid probability distribution on \(\mathcal{X}\)  for every \( u \in \mathcal{U} \), we study the rate--distortion trade-off for lossy compression under the distortion measure
\( d_{\ell\mkern-2mu\ell} : \mathcal{X} \times \mathcal{U} \to [0,\infty] \) defined as
\begin{equation}\label{eq:dist}
d_{\ell\mkern-2mu\ell}(x,y) \triangleq \log \frac{1}{P_{X|U}(x|y)} .
\end{equation}
We refer to $d_{\ell\mkern-2mu\ell}(x,y)$
 as the \emph{log-likelihood loss} (see Fig.~\ref{fig:framework0}).\footnote{Our primary focus in this paper is the rate--distortion trade-off under the proposed log-likelihood loss. While the rate--distortion function characterizes the asymptotic trade-off for lossy compression of i.i.d.\ sources, it is also relevant in the one-shot setting, as shown in \cite{li_strong_2018}. We adopt scalar notation in our figures for simplicity and because it is often more appropriate for semantic compression settings.}
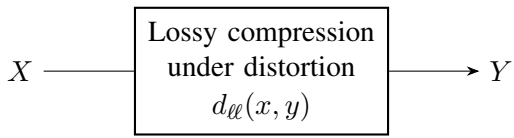
\begin{figure}[H]
  \begin{center}
   
  \begin{tikzpicture}[>=Stealth, baseline,
scale= 1,
transform shape
]

  \node (X) {$X$};

  \node[
    draw=black,
    thick,
    rectangle,
    inner sep=5pt,
    align=center,
    right=12mm of X
  ] (B) {Lossy compression\\under distortion\\[1pt]
         $d_{\ell\mkern-2mu\ell}(x,y)$};

  \node[right=12mm of B] (Y) {$Y$};

  \draw[-] (X) -- (B.west);
  \draw[->] (B.east) -- (Y);

\end{tikzpicture}

       \vspace{2mm}
    \caption{Log-likelihood loss based lossy compression.}
    \label{fig:framework0}
  \end{center}
\end{figure}

We propose this distortion measure to model modern compression settings in which lossy compression serves a dual purpose: producing a compact representation while preserving task-relevant features or semantic information about the source. For instance, \( X \) may represent an image and \( U \) a latent semantic representation describing its content, with \( P_{X|U} \) modeling the probabilistic relationship between the semantics and the image, for example as induced by a trained generative or reconstruction model. Similarly, \( X \) may correspond to a text document and \( U \) to its semantic summary. In such settings, compression under the log-likelihood distortion in \eqref{eq:dist} seeks a representation \( Y \) of limited rate for which the original source realization \( X \) is maximally likely given \( Y \) through the probabilistic mapping \( P_{X|U} \). Equivalently, the distortion quantifies the negative log-likelihood of reconstructing the source from its compressed representation, aligning compression with probabilistic reconstruction fidelity rather than traditional signal-level metrics, which quantify fidelity through a pointwise discrepancy between \( X \) and its reconstruction \( Y \). As such, the log-likelihood loss naturally models scenarios in which the reconstruction \( Y \) does not represent a pointwise approximation of \( X \), but rather an abstract description from which \( X \) can be probabilistically generated.
\begin{figure}[H]
\centering

\begin{subfigure}{\linewidth}
\centering
\usetikzlibrary{positioning}

\centering
\begin{tikzpicture}[
  scale=1,
  transform shape,
  node distance=4.5cm,
  >={Stealth[fill=black,scale=1]},
  midblock/.style={
    draw, rectangle, fill=white, inner sep=3pt, align=center,
    text width=2.8cm, 
    minimum height=1.2cm
  }
]

\node[font=\Large, inner sep=0pt] (Umain) {$U$};
\node[font=\Large, right=of Umain, inner sep=0pt] (Xmain) {$X$};
\node[font=\Large, right=of Xmain, inner sep=0pt] (Ymain) {$Y$};

\draw[->] (Umain) -- (Xmain)
  node[midway, midblock]
  {Noisy\\channel\\$P_{X|U}$};

\draw[->] (Xmain) -- (Ymain)
  node[midway, midblock]
  {Lossy compression\\under distortion\\[1pt]
   $d_{\ell\mkern-2mu\ell}(x,y)$};

\end{tikzpicture}
\vspace{2mm}
\caption{Compression based denoising}
\end{subfigure}
\vspace*{0.5cm}

\begin{subfigure}{\linewidth}
\centering
\begin{tikzpicture}[
  scale=1,
  transform shape,
  node distance=4cm,
    >={Stealth[fill=black,scale=1]}
]

\node[font=\Large, inner sep=0pt] (Umain) {$X$};
\node[font=\Large, right=of Umain, inner sep=0pt] (Xmain) {$Y$};
\node[font=\Large, right=of Xmain, inner sep=0pt] (Ymain) {$\hat{X}$};


\draw[->] (Umain) -- (Xmain)
  node[midway, draw, rectangle, inner sep=5pt, fill=white, align=center]
  {Lossy encoder\\under distortion\\[1pt]
         $d_{\ell\mkern-2mu\ell}(x,y)$};

\draw[->] (Xmain) -- (Ymain)
  node[midway, draw, rectangle, inner sep=5pt, fill=white, align=center]
  {Generative Model\\$P_{X\mid U}$};

\end{tikzpicture}
\label{fig:denoising-setup}

\caption{Compression with generative reconstruction}
\end{subfigure}
\caption{}
    \label{fig:framework1}
\end{figure}

Beyond semantic compression, the proposed distortion measure can be used to model several related but distinct application settings. In Fig.~\ref{fig:framework1}-(a), \( U \) represents an underlying source signal and \( X \) its noisy observation obtained through a channel \( P_{X|U} \). The encoder observes \( X \) and compresses it under the distortion measure in \eqref{eq:dist}. In this setting, the log-likelihood loss promotes denoising through compression. This effect was observed in \cite{Danpaper}, where it was shown that when compression is performed at the average distortion level
\( \mathbb{E}[d_{\ell\mkern-2mu\ell}(X,Y)] = H(X|U) \), the resulting representation serves as a universal denoising of \( X \). A closely related idea appeared earlier in~\cite{Daria}, where the distortion measure \eqref{eq:dist} was used as a cost function for entropic optimal transport in the context of generative modeling from privatized data. 

Finally, Fig.~\ref{fig:framework1}-(b) illustrates a complementary application in which the decoder is fixed in advance to a given probabilistic reconstruction model \( P_{X|U} \). In this setting, the source is compressed with the knowledge that reconstruction will be performed probabilistically according to \( P_{X|U} \), for example via a generative or AI-based model. The distortion measure \eqref{eq:dist} ensures that the compression strategy is matched to the decoder in the sense of maximizing the likelihood of reproducing the original source.

While these scenarios represent distinct applications of the proposed log-likelihood distortion measure—ranging from semantic compression to denoising and fixed-decoder reconstruction—they collectively highlight its relevance as a natural distortion measure for modern compression problems involving semantic representations.

\subsection{Related work} Semantic compression has gained traction in recent years and several models have been proposed for semantic-preserving lossy compression. In particular, foundation-model based semantic compression techniques have been studied in~\cite{f1,f2,f3}. An approach to semantic compression based on information lattice learning was proposed in~\cite{lav}.  In~\cite{sac} authors study an arithmetic coding based method for semantic lossless compression.~\cite{srdt} studied a rate-distortion framework for semantic compression via divergence measure inspired by information-bottleneck constraint i.e., based on distribution of semantics $(U)$ conditioned on the observation $(X)$ and the reconstruction $(Y)$ i.e., $P_{U|X}$ and $P_{U|Y}$ respectively, along with an observation distortion measure. Other frameworks for semantic compression based on rate-distortion have been proposed under different metrics in~\cite{srd1,srd2,srdp1}. Our approach significantly deviates from these earlier approaches as we capture semantics via a simple and intuitive novel distortion measure inspired by log-loss, under which the goal is to compress $X$ into $Y$ such that $X$ remains highly likely under $P_{X|U}(\cdot|Y)$.

\subsection{Contributions and Organization} In this paper, we focus on studying the rate-distortion function under the proposed log-likelihood loss. In section~\ref{sec:properties}, we present several properties of the rate-distortion function  and explore its relations with the standard log-loss distortion. Though, the log-likelihood loss in~\eqref{eq:dist} looks deceptively restricted, in section~\ref{sec:classical-rd} we show that it generalizes several commonly studied classical rate-distortion frameworks indicating its applicability across various rate-distortion scenarios, followed by an illustrative example. In section~\ref{sec:rdp}, we show that our framework provides an achievable scheme to attain rate-distortion with perfect perception for a special class of rate-distortion problems. The detailed proofs are deferred to the Appendix.

\section{Notations and Background}

The PMF of a discrete random variable (PDF for continuous random variables) $X$ is denoted using a upper case letter, say $P_X$, while the probability of an event is denoted using the bold-face letter $\mathbb{P}$. Given a random variable $X$, its support (and sets in general) is denoted by $\cX$, while a realization is denoted by lower case letter, for example, $x\in \cX$. For a joint distribution $P_{XYZ}$, the $[P_{XYZ}]_{XY}$ denotes the joint distribution of $(X,Y)$. We use $\Delta(\mathcal{X})$ to denote the simplex on $\mathcal{X}$. We denote the expectation of the random variable $X$ as bold-face $\mathbb{E}[X]$. The set of all real numbers and non-negative real numbers are denoted by $\mathbb{R}$ and $\mathbb{R}_{+}$, respectively.
We use $H(X)$ denote the Shannon entropy (differential entropy for continuous random variables) of a random variable $X$. All logarithms are to the base $e$, unless stated otherwise.
\begin{definition}[Rate-Distortion Function (RDF)] Given an information source $X \sim P_X$ taking values in $\mathcal{X}$. Let $Y \in \mathcal {Y}$ be the lossy reconstruction of $X$ under the distortion measure
$d:\mathcal X\times{\mathcal{Y}}\rightarrow [0,\infty]$. Then, the rate-distortion function (RDF) for $(X,d)$ is given by
\begin{align}
R(D) =\min_{W_{Y|X}:\ \mathbb{E}_{XY}[d(X,Y)]\leq D}\ I(X;Y)
\end{align}
for any $D \geq 0$.
\end{definition}
\begin{definition}[Rate-Distortion Function under Log-Likelihood Loss (RLLDF)]\label{def:srd}
Consider the lossy compression of $X$ into a reconstruction $Y \in \mathcal{U}$ under the log-likelihood distortion measure $d_{\ell\mkern-2mu\ell}: \mathcal{X} \text{ } \times \text{ } \mathcal{U} \rightarrow [0,\infty]$ such that $d_{\ell\mkern-2mu\ell}(x,y):=-\log P_{X|U}(x|y)$, where $P_{X|U}(x|u)$ for $x\in \mathcal{X}$ and $u\in\mathcal{U}$ is a given conditional distribution. The rate-distortion function under log-likelihood loss for $(X, P_{X|U})$ is defined as follows
\begin{align}
    R_{\ell\mkern-2mu\ell}(D)=  \min_{W_{Y|X}:\mathbb{E}_{XY}[d_{\ell\mkern-2mu\ell}(X,Y)] \leq D}I(X;Y)
\end{align}
\end{definition}

\begin{remark}
 $R_{\ell\mkern-2mu\ell}(D)$ can also be expressed as
\begin{align}
    R_{\ell\mkern-2mu\ell}(D)&= \min_{\substack{W_{Y|X}:\\H(X|Y)+\mathbb{E}_{Y}[D_\mathrm{KL}(W_{X|Y}(\cdot|Y)||P_{X|U}(\cdot|Y))] \leq D}}I(X;Y) \label{eq:comp:srd}\\
    &=\min_{\substack{W_{Y|X}:\\\mathbb{E}_{Y}[\mathrm{CE}(W_{X|Y}(\cdot|Y)||P_{X|U}(\cdot|Y))] \leq D}}I(X;Y)
\end{align}
where $\mathrm{CE}(\cdot\|\cdot)$ denotes the cross-entropy function.
\end{remark}

\section{Properties OF \texorpdfstring{$R_{\ell\mkern-2mu\ell}(D)$}{R[ll](D)}}\label{sec:properties}
In the following, we state some properties of the  RDF under log-likelihood loss  and it's connections with other rate-distortion problems studied in the literature.
\begin{theorem}\label{thm:prop1}
For given $(X, P_{X|U})$, $R_{\ell\mkern-2mu\ell}(D)$ is defined only for the distortion values $D\in [D_{\min},D_{\max}]$ where
\begin{align}
    D_{\min} &\triangleq \mathbb{E}_X\left[\min_u \log\frac{1}{P_{X|U}(X|u)}\right]\\
    D_{\max} &\triangleq \min_u \mathbb{E}_X\left[\log\frac{1}{P_{X|U}(X|u)}\right]
\end{align}
Moreover, the following facts follow.
\smallbreak
\begin{itemize}
  \item[(i)] At distortion $D=D_{\min}$,   we have
\begin{align}\label{eq:Dmin}
    R_{\ell\mkern-2mu\ell}(D_{\min}) = \min_{Q \in \Delta(\mathcal{U})} \mathbb{E}_X\left[-\log \sum_{u \in T(X)}Q_U(u) \right]
\end{align}
where $T(x):= \{u \in \mathcal{U}: u= \arg\max_{u'}P_{X|U}(x|u')\}$. The optimal reconstruction is a {randomized maximum-likelihood (ML) decoder}: for each $x$, the decoder outputs a random $y \in T(x)$, where the optimal tie-randomization strategy  over $T(x)$ is given by the minimizer of \eqref{eq:Dmin}.\\ 
\item[(ii)] At distortion $D= D_{\max}$, we have $R_{\ell\mkern-2mu\ell}(D_{\max})=0$.\\
\item[(iii)] Assume there exists $(U,X)\sim P_{U,X}$ which is consistent with $(X, P_{X|U})$. At distortion $D^{*}=H(X|U)$, $R_{\ell\mkern-2mu\ell}(D)$ admits a closed-form expression i.e.,
\begin{align*}
    R_{\ell\mkern-2mu\ell}(D^{*}) = I(U;X)
\end{align*}
We refer to the distortion level $D^{*}$ as a special operating point for $(X, P_{X|U})$. At distortion $D=D^*$, the reconstruction $Y$ of the optimal compressor is a sample from the posterior $P_{U|X}$, which further implies $Y \sim P_U$.
\end{itemize}
\end{theorem}
\begin{proof}
    The proof is deferred to the appendix~\ref{app:prop}.
\end{proof}

\subsection{Connection to Log-loss Distortion}
Rate distortion under log-loss has been studied widely in the literature~\cite{courtade,pmlr_shkel,shkel_it}. For a given source $X \sim P_X$, the decoder outputs a predictive distribution $W \in \Delta(\mathcal{X})$ under the log-loss distortion $d(x,W) := -\log W(x)$. At a given distortion level $D$, the log-loss RDF is given by
\begin{align}
   R_{\ell}(D) &=\min_{W_{Y|X}: H(X|Y) \leq D} I(X;Y)\label{eq:comp:llrd}\\
   &= H(X) -D 
\end{align}
where $D \in [0,H(X)]$.

The log-likelihood loss can also be interpreted as allowing the decoder to output a predictive distribution, however in this case the predictive distribution is
restricted to a strict subset of the simplex, namely the family
$\{P_{X|U}(\cdot|u)\}_{u\in\mathcal{U}} \subset \Delta(\mathcal{X})$. Concretely, if the decoder outputs a symbol $y \in \mathcal{U}$, this induces the predictive distribution $W_y := P_{X|U}(\cdot|y)$, and the log-likelihood loss becomes
\begin{align}
d_{\ell\mkern-2mu\ell}(x,y) = -\log P_{X|U}(x|y) = d(x,W_y),
\end{align}

Hence the log-likelihood loss can be viewed as a restricted version of the log-loss framework where the decoder is allowed to output a certain class of predictive distributions. As such, the trade-off under log-loss, not surprisingly, serves as a lower bound for the trade-off under log-likelihood loss as stated in the next theorem. From a practical perspective, this can model  settings where the decoder is indeed restricted to use a given predictive model as illustrated in Fig.~2-(b).
\begin{figure}
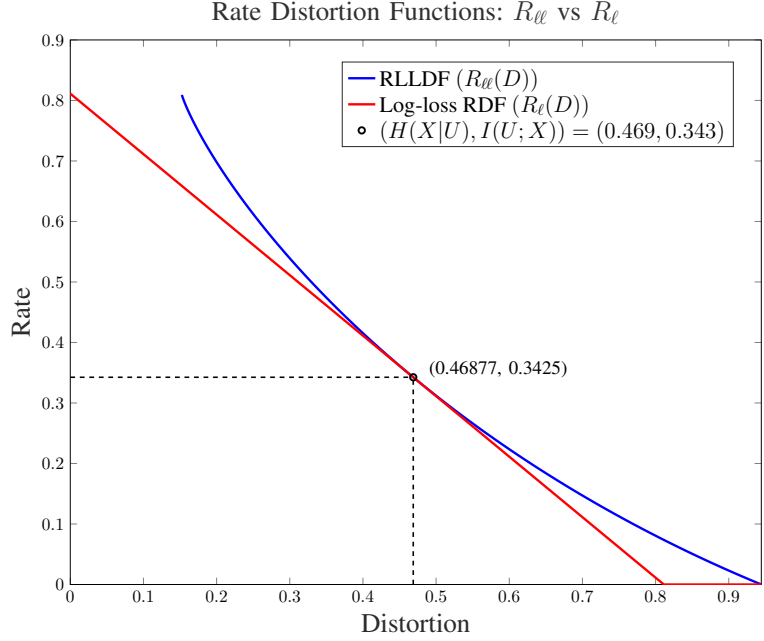

  \centering

%
  \caption{$(X,P_{X|U})$:  $X \sim \mathrm{Ber}(0.25)$ and $P_{X|U} \sim \mathrm{BSC}(0.1)$, the figure plots $R_{\ell}(D)$ for $D \in [0,H(X)=0.8113]$ and $R_{\ell\mkern-2mu\ell}(D)$ for $D \in [D_{\min} =0.152, D_{\max}=0.945]$. At $D^{*}= 0.469$, both RDFs coincide. $D^{*}$ is the special operating point for $(X,P_{X|U})$ with $U \sim \mathrm{Ber}(0.1875)$.}\label{fig:example}
\end{figure}
The relationship between the log-likelihood loss and the standard log-loss can be further understood by comparing the corresponding distortion constraints in \eqref{eq:comp:srd} and \eqref{eq:comp:llrd}. Under the log-loss distortion in \eqref{eq:comp:llrd}, the distortion constraint bounds only the conditional entropy \( H(X|Y) \), thereby enforcing that the representation \( Y \) is predictive of \( X \). In contrast, with the log-likelihood distortion, we bound \( H(X|Y) \) together with an additional KL divergence term that measures the discrepancy between the predictive distribution induced by \( Y \) and the prescribed probabilistic model \( P_{X|U} \).

In addition to modeling a constrained decoder, this additional term can be interpreted as guiding the compression toward a desired notion of semantics. With log-loss alone, there is no restriction on the form of the predictive distribution, and therefore the distortion does not distinguish, for example in the case of images, between a coarse quantization of the pixels and a semantic representation, provided both are equally predictive of the source image. In contrast, when using the log-likelihood distortion, a prescribed semantic structure encoded through \( P_{X|U} \) allows the compression to be steered toward representations that align with that notion of semantics.
\begin{figure}
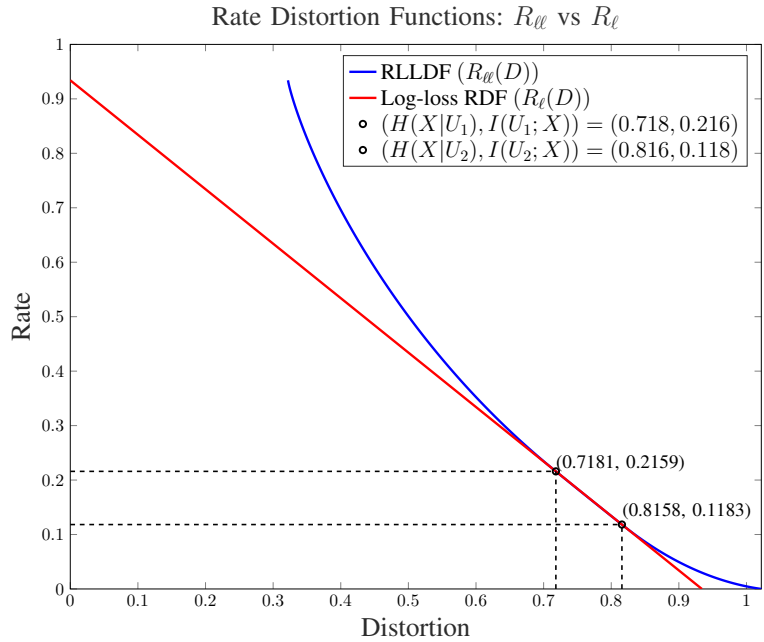

  \centering
%
%
\definecolor{mycolor1}{rgb}{0.00000,0.44700,0.74100}%
\definecolor{mycolor2}{rgb}{0.85000,0.32500,0.09800}%
%
  \caption{$(X,P_{X|U})$: $X \sim \mathrm{Ber}(0.35)$ and $P_{X|U} = [
0.8,  0.4,  0.2;
0.2, 0.6, 0.8]
$, the figure plots $R_{\ell}(D)$ for $D \in [0,0.934]$ and $R_{\ell\mkern-2mu\ell}(D)$ for $D \in [0.322, 1.022]$. Both RDFs coincide for $D^{*} \in [0.718,0.816]$. Every $D^{*}$ in this interval is a special operating point with a corresponding $(U_i,X)$ consistent with $(X,P_{X|U})$.Since RLLDF $R_{\ell\mkern-2mu\ell}(D)$ does not exhibits a closed-form expression for this $(X,P_{X|U})$, the $R_{\ell\mkern-2mu\ell}(D)$ is plotted using Blahut-Arimoto algorithm~\cite{blahut,arimoto}.}\label{fig:example2}
\end{figure}
\begin{theorem}\label{thm:logloss}
Let $R_{\ell}(D)$ be the RDF under log-loss for $X$, and let $R_{\ell\mkern-2mu\ell}(D)$ be the RDF under log-likelihood loss for $(X, P_{X|U})$. Then, for all feasible $D \in [D_{\min},D_{\max}]$,
\begin{align}
R_{\ell\mkern-2mu\ell}(D)\ge R_{{\ell}}(D)
\end{align}
where the equality holds for all special operating points as defined in Theorem~\ref{thm:prop1}-(iii), i.e. whenever there exists  $(U,X)\sim P_{U,X}$ which is consistent with $(X, P_{X|U})$ and $D=H(X|U)$.
\end{theorem}
\begin{proof}
    The proof is deferred to the appendix~\ref{app:logloss}.
\end{proof}
\begin{corollary}
    Assume there exists $(U_1,X)\sim  P_{U_1,X}$ and $(U_2,X)\sim  P_{U_2,X}$ s.t. both $ P_{U_1,X}$ and $ P_{U_2,X}$ are consistent with $(X, P_{X|U})$. 
    Let $D_1=H(X|U_1)$ and $D_2=H(X|U_2)$, then for all $D_1\leq D\leq D_2$,
    $$
    R_{\ell\mkern-2mu\ell}(D)= H(X)-D.
    $$
\end{corollary}
The corollary follows as a consequence of the fact that $(X, P_{X|U})$ has multiple special operating points and $R_{\ell\mkern-2mu\ell}(D)$ touches the linear log-loss trade-off at all these points. In Fig.~\ref{fig:example2} we provide an example where this is the case.

\section{Connection to General Rate-Distortion Problems}\label{sec:classical-rd}
In this section, we establish correspondence between lossy compression of a source $X\sim P_X$ under a general distortion measure $d$ and log-likelihood loss. Given a source $X~\sim P_X$ and a distortion measure $d:\mathcal X\times{\mathcal{Y}}\rightarrow [0,\infty]$. Suppose there exists a $\lambda>0$ and a nonnegative function
$\mu(x,\lambda)$, independent of $y$, such that for every
$y\in\mathcal{Y}$ the following defines a valid conditional distribution
\begin{equation}
    P^{\lambda}_{X|Y}(x|y)= \mu(x,\lambda)e^{-\lambda d(x,y)},\label{eq:condition}
\end{equation}
Then the rate-distortion problem for $(X,d)$ at distortion level $D$ can be reformulated as a log-likelihood loss problem $(X,P_{X|U})$ at distortion
$\widetilde{D}$ related to $D$ via an affine mapping.

Given such a $(X,d)$ satisfying~\eqref{eq:condition}. Now consider compressing the same source $X\sim P_X$ under log likelihood loss, i.e. the problem $(X, P_{X|U})$ with  $P_{X|U}$ chosen to coincide with $P^{\lambda}_{X|Y}$ above. In other words, we have
    \begin{align}
        d_{\ell\mkern-2mu\ell}(x,y) &= -\log P_{X|U}(x|y)\\
        &= - \log \Big(\mu(x,\lambda)e^{-\lambda d(x,y)}\Big)\\
        &= \lambda d(x,y) - \log(\mu(x,\lambda)).
    \end{align}
By taking expectation on both sides w.r.t $(X,Y)$, we have
   \begin{align}
       \mathbb{E}_{XY}[d_{\ell\mkern-2mu\ell}(X,Y)] = \lambda \mathbb{E}_{XY}[d(X,Y)] - \mathbb{E}_X[\log(\mu(X,\lambda))]
   \end{align}
Then we have 
\begin{align}
       R_{\ell\mkern-2mu\ell}(\widetilde{D}) &=\min_{W_{Y|X}:\mathbb{E}_{XY}[d_{\ell\mkern-2mu\ell}(X,Y)]\leq \widetilde{D}}I(X;Y)\\
       &= \min_{W_{Y|X}:\mathbb{E}_{XY}[d(X,Y)] \leq \frac{\widetilde{D}+\mathbb{E}_X[\log(\mu(X,\lambda))]}{\lambda}}I(X;Y)\\
       &=R(D) \label{eq:equivalence}
\end{align}  
where we chose 
\begin{align}
\widetilde{D} = \lambda D -\mathbb{E}_X[\log(\mu(X,\lambda))].\label{eq:choice}
\end{align}
In other words, if $P^{\lambda}_{X|Y}$ in \eqref{eq:condition} exists then $(X,d)$  and $(X, P_{X|U})$ with $P_{X|U}=P^{\lambda}_{X|Y}$ are equivalent in the sense that the corresponding RDFs are related through an affine transformation of the distortion value.  This shows that the log-likelihood distortion problem $(X,P_{X|U})$  is as general as $(X,d)$. Therefore, we also cannot expect to characterize the RDF for $(X,P_{X|U})$ in full generality unless we can characterize the RDF for $(X,d)$ for any $X$ and $d$.

Therefore, in the sequel we focus on a class of rate-distortion problems $(X,d)$ for for which the rate-distortion function can be computed by solving a single parameter optimization problem.

\begin{theorem}\label{thm:dual}
    Given an information source $X~\sim P_X$ with entropy $H(X)$ and a distortion measure $d(\cdot,\cdot)$, let $\mathcal{I}$ denote the set of all $\lambda>0$  such that there exists 
     \begin{enumerate}[(i)]
     \item a real-valued function $\mu(\cdot, \lambda): \mathcal{X} \rightarrow [0,\infty)$
        \item a coupling $P^{\lambda}_{Y,X}$ on $(\mathcal{Y} \times \mathcal{X})$, such that $\forall$ $(y,x)$ we have $P^{\lambda}_{X|Y}(x|y)= \mu(x,\lambda)e^{-\lambda d(x,y)}$.
    \end{enumerate}
If $\mathcal{I}$ is non-empty, i.e., $\mathcal{I} \neq \varnothing$, the RDF $R(D)$ for $(X,d)$ can be 
expressed as a single-parameter optimization problem i.e.,
\begin{align}
    R(D) = \max_{\lambda \in \mathcal{I}} \Big(H(X) + \mathbb{E}_X[\log(\mu(X,\lambda))] -\lambda D \Big)\label{eq:single}
\end{align}
\end{theorem}
\begin{proof}
The proof is deferred to the Appendix~\ref{app:dual}. It is based on solving the dual form for $R(D)$  through the Lagrangian form and applying KKT conditions while invoking the feasibility conditions to restrict the optimization to $\lambda\in \mathcal{I}$. The dual form of $R(D)$ also appears in~\cite{kostina,gallager}.
\end{proof}

\begin{remark}
    Theorem~\ref{thm:dual} characterizes a class of rate–distortion problems for which a `generalized Shannon-type lower bound' is tight, yielding the single-parameter representation in~\eqref{eq:single}. In the special case of continuous sources with difference distortions, this tight bound reduces to the classical Shannon lower bound~\cite{shannon_lb}.
\end{remark}

Note that the conditions of Theorem~\ref{thm:dual} are stronger than the condition in \eqref{eq:condition} as they also require the existence of  $Y\sim Q_Y$ (equivalently, a coupling $P^\lambda_{X,Y}$ with $P^\lambda_X = P_X$) such that a single function $\mu(x,\lambda)$ simultaneously normalizes $P^\lambda_{X|Y}(\cdot|y)$ for all $y$. The set $\mathcal{I}$  is a collection of those values of $\lambda$ for which such a coupling exists; accordingly, the maximization in~\eqref{eq:single} is restricted to $\lambda\in\mathcal{I}$.

Since the conditions of Theorem~\ref{thm:dual} imply \eqref{eq:condition}, the RDF for the corresponding $(X, P_{X|U})$ with $P_{X|U}=P^{\lambda}_{X|Y}$  can also be expressed in the form  \eqref{eq:single} through the translation in \eqref{eq:equivalence}. This gives us a family of log-likelihood loss problems for which we can characterize the corresponding RDF as we illustrate in the following examples.
\begin{example} [Binary Source with Hamming Distortion]\label{ex:a1} Let $X \sim \mathrm{Ber}(p)$, where $p \in [0,1/2]$, and the distortion measure is the hamming distortion i.e.,
\begin{align}
d_H (x,y)=
\begin{cases}
    0 \text{ }; \text{ } \text{if} \text{ } x=y\\
    1 \text{ } ; \text{ } \text{if} \text{ } x \neq y\\
\end{cases}
\end{align}
\end{example}
\noindent\emph{Theorem~\ref{thm:dual} (Existence of coupling):} Fix a $\lambda > 0$. Then, $P^{\lambda}_{X|Y}$ (due to condition (ii)) is given by 
\begin{align}
P^{\lambda}_{X|Y} (x|y) =
\begin{bmatrix}
\dfrac{1}{1+e^{-\lambda}} & \dfrac{e^{-\lambda}}{1+e^{-\lambda}} \\[2pt]
\dfrac{e^{-\lambda}}{1+e^{-\lambda}} & \dfrac{1}{1+e^{-\lambda}}
\end{bmatrix}
\end{align}
where $\mu(0,\lambda)=\mu(1,\lambda)=\frac{1}{1+e^{-\lambda}}$. Let $Y~\sim \mathrm{Ber}(q)$ for some $q$. We will show that there exists a $\lambda > 0$ which guarantees the existence of a coupling. Using Bayes rule, $P^{\lambda}_{Y|X}$ is
\begin{align}
P^{\lambda}_{Y|X}(y| x)=
\begin{bmatrix}
\dfrac{\left(\frac{1}{1+e^{-\lambda}}\right)(1-q)}{1-p} &
\dfrac{\left(\frac{e^{-\lambda}}{1+e^{-\lambda}}\right)q}{1-p}
\\[2pt]
\dfrac{\left(\frac{e^{-\lambda}}{1+e^{-\lambda}}\right)(1-q)}{p} &
\dfrac{\left(\frac{1}{1+e^{-\lambda}}\right)q}{p}
\end{bmatrix}
\end{align}
We see that the rows of $P^{\lambda}_{Y|X}$ sum to $1$ iff 
\begin{align}
q
=\frac{p(1+e^{\lambda})-1}{e^{\lambda}-1}
\end{align}
For $q \in [0,1]$, we have that $\lambda \geq \log(\frac{1-p}{p})$. Thus, given $p \in [0,1/2]$ we observe that a valid coupling $P^{\lambda}_{XY}$ exists for any $\lambda \in \mathcal{I}:=[\log(\frac{1-p}{p}),\infty)$. Therefore, from~\eqref{eq:single} 
\begin{align}
    R(D) &= \max_{\lambda \in \mathcal{I}}\Big(H(p)+\log\frac{1}{1+e^{-\lambda}}-\lambda D\Big)\\
    &=H(p) + \log(1-D)-D\log\Big(\frac{1-D}{D}\Big)\\
    &= H(p)-H(D)
\end{align}
where the maximizer $\lambda^{*}=\log(\frac{1-D}{D}) \in \mathcal{I}$.
\smallbreak
\noindent\emph{Translation into log-likelihood loss $(X,P_{X|U})$}:
Fix a $\lambda_{0} \in \mathcal{I}$. Now, choose $P_{X|U}$ in $(X,P_{X|U})$ to be $P^{\lambda_0}_{X|Y}$. Thus,
\begin{align}
P^{}_{X|U} (x|u) =
\begin{bmatrix}
\dfrac{1}{1+e^{-\lambda_0}} & \dfrac{e^{-\lambda_0}}{1+e^{-\lambda_0}} \\[2pt]
\dfrac{e^{-\lambda_0}}{1+e^{-\lambda_0}} & \dfrac{1}{1+e^{-\lambda_0}}
\end{bmatrix}
\end{align}
From \eqref{eq:choice} we have
\begin{align}
    \widetilde{D}
    &= \lambda_0 D -\mathbb{E}_{X}\left[\log\left(\frac{1}{1+e^{-\lambda_0}}\right)\right]\\
    &=\lambda_0 D + \log(1+e^{-\lambda_0})
\end{align}
Therefore, we obtain that 
\begin{align}
    R_{\ell\mkern-2mu\ell}(\widetilde{D}) = H(p)-H\left(\frac{\widetilde{D}-\log(1+e^{-\lambda_0})}{\lambda_0}\right)
\end{align}
where $\widetilde{D} \in [\log(1+e^{-\lambda_0}),\lambda_0p +\log(1+e^{-\lambda_0})]$.\\

\begin{example} [Gaussian source with squared distortion]\label{ex:a2}Let $X \sim \mathcal{N}(0,\sigma^2)$ and the distortion measure be  the squared error distortion i.e., $d(x,y)=|x-y|^2$.
\end{example}
\noindent\emph{Theorem~\ref{thm:dual} (Existence of coupling): }Fix a $\lambda > 0$. Let the conditional distribution $P^{\lambda}_{X|Y}$ be as in condition (ii) of Theorem~\ref{thm:dual} i.e.,
\begin{align}
  P^{\lambda}_{X|Y}(x|y)=\mu(x,\lambda)e^{-\lambda(x-y)^2}  
\end{align}
$P^{\lambda}_{X|Y}$ is a valid conditional distribution for all $(x,y)$ iff $\mu(x,\lambda) = \sqrt{\lambda/\pi}$. Thus,
\begin{align}
  P^{\lambda}_{X|Y}(x|y)=\sqrt{\frac{\lambda}{\pi}}e^{-\lambda(x-y)^2}  \label{eq:laplacec}
\end{align}
Assume that a valid coupling exists, we will now show the existence of a consistent PDF of $Y$ using  use the approach of characteristic functions. Let $\phi_X(t)=\mathbb{E}[e^{itX}]$ and $\phi_Y(t)=\mathbb{E}[e^{itY}]$ denote the characteristic functions of $X$ and $Y$ respectively. On using towering property, we have
\begin{align}
\phi_X(t)=\mathbb{E}\left[e^{itX}\right]
      =\mathbb{E}\left[\mathbb{E}\left[e^{itX}\mid Y\right]\right]
\end{align}
We are given $Y=y$, $X=y+Z$ where $Z$ has the gaussian density $\sqrt{\frac{\lambda}{\pi}}e^{-\lambda|z|^2}$, Thus, we have
\begin{align}
\mathbb{E}\left[e^{itX}\mid Y=y\right]&=e^{ity}\phi_Z(t) \\
\phi_Y(t)&=\frac{\phi_X(t)}{\phi_Z(t)} \label{eq:char2}
\end{align}
Since $\phi_X(t)=e^{-\sigma^2 t^2/2}$ and $\phi_Z(t)=e^{- t^2/4\lambda}$, substituting back into~\eqref{eq:char2} gives
\begin{align}
    \phi_Y(t)&= \exp\left(-\frac{\sigma^2 t^2}{2}+\frac{t^2}{4\lambda}\right)\\
    &=\exp\left( -\frac{t^2}{2}\left( \sigma^2 - \frac{1}{2\lambda}\right)\right)
\end{align}
Using Fourier inversion, we see that $Y$ has a gaussian density i.e., $Y \sim \mathcal{N}\left(0,\left(\sigma^2-\frac{1}{2\lambda}\right)\right)$ and consequently $W_{Y|X} \sim \mathrm{AWGN}\left(\frac{1}{2\lambda}\right)$. Thus, a valid coupling (and a valid marginal $Y$) exists for every $\lambda > \frac{1}{2 \sigma^2}$. Thus, $\mathcal{I} =\left(\frac{1}{2 \sigma^2},\infty\right)$.

We can now use the single parameter optimization to obtain $R(D)$ for $(X,d)$,
\begin{align}
    R(D) &= \max_{\lambda \in \mathcal{I}}\Big(H(X)+\log\sqrt{\frac{\lambda}{\pi}}-\lambda D\Big)\\
    &=\frac{1}{2}\log(2\pi e\sigma^2) + \frac{1}{2}\log\frac{1}{2\pi D}-\frac{1}{2}\\
    &= \frac{1}{2}\log \left(\frac{\sigma^2}{D} \right)
\end{align}
where the maximizer $\lambda^{*}= 1/2D \in \mathcal{I}$.\\

\noindent\emph{Translation into log-likelihood loss $(X,P_{X|U})$}:  Fix any $\lambda_0 \in \mathcal{I}$. Define the channel $P_{X|U}$ to be the backward test channel for $(X,d)$ at $\lambda_0$. Then,
\begin{align}
P^{}_{X|U} (x|u) =
\sqrt{\frac{\lambda_0}{\pi}}e^{-\lambda_0(x-u)^2} 
\end{align}
From \eqref{eq:choice} we have
\begin{align}
    \widetilde{D}
    &= \lambda_0 D -\mathbb{E}_{X}\left[\log\sqrt{\lambda_0/\pi}\right]\\
    &=\lambda_0 D - \log (\sqrt{\lambda_0/\pi})
\end{align}
Therefore, we obtain that 
\begin{align}
    R_{\ell\mkern-2mu\ell}(\widetilde{D}) = \frac{1}{2}\log \left(\frac{\sigma^2 \lambda_0}{\widetilde{D}+\log(\sqrt{\lambda_0/\pi})} \right)
\end{align}
where $\widetilde{D} \in [\log(\sqrt{\pi/\lambda_0}),\sigma^2\lambda_0 +\log(\sqrt{\pi/\lambda_0})]$.\\

There are several other classical rate-distortion problems which can be translated into rate-distortion with log-loss framework $(X,P_{X|U})$, such as binary source with asymmetric hamming distortion, Laplace source with absolute error distortion, etc. However, not all of them satisfy the conditions in Theorem~\ref{thm:dual}. One such example is the Gaussian source with absolute error distortion.

\section{Connection to Rate-Distortion-Perception}\label{sec:rdp}
In this section, we show that the log-likelihood rate distortion framework provides an achievable scheme for attaining rate distortion with prefect perception.
\begin{definition}[Rate Distortion-Perception~\cite{rdp,rdp2}]
Given an information source $X \sim P_X$ taking values in $\mathcal{X}$. Let $Y \in \mathcal {X}$ be the lossy reconstruction of $X$ under the distortion measure
$d:\mathcal X\times{\mathcal{Y}}\rightarrow [0,\infty]$ and perception measure $\gamma:\Delta(\mathcal{X}) \times \Delta(\mathcal{X})\rightarrow [0,\infty]$. We assume $\gamma(P, Q) = 0$ iff $P=Q$. Then, the rate distortion perception function for $(X,d,\gamma)$ is given by
\begin{align}\label{eq:rdp-perfect1}
R(D,\Gamma) =\min_{\substack{
W_{Y|X}:\\
\mathbb{E}_{XY}[d(X,Y)]\leq D, \text{ }
\gamma(P_X;Q_Y) \leq \Gamma
}}  I(X;Y)
\end{align}
for any $D,\Gamma \geq 0$. At $\Gamma =0$, we say that perfect perception is achieved.
\end{definition}
Next, we define the Completely Positive (CP) matrices which will be later used to state our result in Theorem~\ref{thm:perc}.
\begin{definition}[Completely Positive (CP) matrix~\cite{berman_completely_2003}]\label{def:cp}

A matrix $A \in \mathbb{R}^{n \times n}$ is called \emph{completely positive} if there exist an integer $r \ge 1$ and a matrix
$B \in \mathbb{R}_+^{n \times r}$ (i.e., $B_{ij} \ge 0$ for all $i,j$) such that
\begin{equation}
A = B B^\top .
\end{equation}
Equivalently, $A$ is the Gram matrix of a finite collection of vectors in $\mathbb{R}_+^n$:
$A_{ij} = \langle b_i, b_j \rangle$ with $b_i \in \mathbb{R}_+^r$.
\end{definition}

From Theorem~\ref{thm:prop1} (iii), given $(X, P_{X|U})$ if there exists $(U,X)$ consistent with $(X, P_{X|U})$ and $D=H(X|U)$, then the optimal reconstruction $Y$ is such that $P_Y = P_U$. This observation can be translated to an achievable strategy for the rate-distortion problem with perfect perception as follows. Given $X$ assume we choose any $P_{Z|X}(z|x)$ and process $X$ into $Z$ according to $P_{Z|X}$. Then, we compress $Z$ under the log-likelihood distortion $d_{\ell\mkern-2mu\ell}(z,y) = -\log P_{Z|X}(z|y)$ at distortion level $H(Z|X)$ which ensures that the reconstruction $Y$ has the distribution $P_X$. 

End to end, this yields an achievable scheme for the rate-distortion-perception problem with perfect perception. In other words, perfect perception can be achieved by first taking $X$ and passing it through an arbitrary noisy channel $P_{Z|X}(z|x)$ and then lossily compressing the noisy observation under the log-likelihood distortion induced by the noisy channel at the fixed distortion level $H(Z|X)$.\\

In \cite{rdp}, a similar construction is used in the case of MSE distortion to upper bound $R(D,0)$. Below we characterize when the above scheme based on log-likelihood loss is optimal.

\begin{theorem}\label{thm:perc}
Let $(X,d,\gamma)$ describe a rate-distortion-perception problem such that $|\mathcal{X}|<\infty$ and $d(\cdot,\cdot)$ is such that the element-wise exponential matrix $V$, defined by $V(x,y):=\exp_H{(-\lambda d(x,y))}$ is completely positive (CP), for every $\lambda>0$.
Then, there exists a $Z$ such that the lossy compression of $(Z,P_{Z|X})$ under log-likelihood loss at $H(Z|X)$ induces a coupling $[W_{XZY}]_{XY}$ which achieves rate-distortion with zero perception error i.e.,
\begin{align}
\label{eq:rdp-perfect}
R(D,0) &= \min_{P_{Z|X}:} I(X;Y)\\
&\text{subject to }\bbE_{XY}\left[d(X,Y)\right]\leq D,\\
& W_{X,Z,Y}=P_XP_{Z|X}P_{X|Z}
\end{align}
\end{theorem}
\begin{proof}
    The Proof is deferred to the Appendix~\ref{app:perception}.
\end{proof}

Let $W^{*}_{Y|X}$ be the optimal conditional distribution achieving rate-distortion with zero perception error in~\eqref{eq:rdp-perfect1} i.e., $R(D,0)$, which induces the optimal joint coupling $W^{*}_{X,Y}$. Now, if the optimal coupling $W^{*}_{X,Y}$ is a CP matrix then, there exists a $Z$ such that $X - Z -Y$ and $P_{X|Z} = P_{Y|Z}$, and vice-versa. Then we can choose the induced $P_{Z|X}$ and process $X$ into $Z$ according to $P_{Z|X}$. Then, with $(Z, P_{Z|X})$ we compress $Z$ under the log-likelihood distortion $d_{\ell\mkern-2mu\ell}(z,y) = -\log P_{Z|X}(z|y)$ at the distortion level $H(Z|X)$ which ensures that the reconstruction $Y$ has the distribution $P_X$. Moreover, the optimal joint distribution $W^{*}_{X,Z,Y}$ induces the joint distribution on  $(X,Y)$ i.e., $[W^{*}_{X,Z,Y}]_{X,Y}$ which achieves the optimal rate at zero perception error.

\begin{remark}
    The condition that $V$ is CP is satisfied for many reasonable rate distortion problems.
    For example, it is satisfied when $d$ is squared distance on a finite subset of $\mathbb{R}^n$ as well as the Hamming distortion measure (see Appendix~\ref{app:sufficient} for details).
\end{remark}

\newpage
\bibliographystyle{IEEEtran}
\bibliography{IEEEabrv,ref}
\newpage

\appendix

\subsection{Proof of Theorem~\ref{thm:prop1}}\label{app:prop}
We give the proof for finite alphabets.
The proof extends to continuous alphabets by replacing $P_{X|U}$ with the appropriate density functions.
Let $(X,P_{X|U})$ describe the log-likelihood based rate distortion. Then, the rate distortion function is given by
\begin{align}
    R_{\ell\mkern-2mu\ell}(D)&=  \min_{W_{Y|X}:\mathbb{E}_{XY}[-\log P_{X|U}(X|Y)] \leq D}I(X;Y).
    \end{align}

\noindent\emph{Proof for $D_{\min}$ and $R(D_{\min})$:}
Given $(X,P_{X|U})$, we will prove a `realizable' lower bound on the expected distortion, which gives us $D_{\min}$. Thus, we have
\begin{align}
   \mathbb{E}_{XY}[-\log P_{X|U}(X|Y)] &= \mathbb{E}_X\Bigg[\bbE_Y\Big[\log\frac{1}{P_{X|U}(X|Y)}\Big| X\Big]\Bigg]\\
   & \geq \mathbb{E}_X\Bigg[\min_{y \in \mathcal{U}}\log\frac{1}{P_{X|U}(X|y)}\Bigg]\\
   &:=D_{\min}
\end{align}
Thus, the distortion constraint is satisfied with equality at $D= D_{\min}$ if and only if
\begin{align}\label{eq:sss}
    \bbE_Y\Big[\log\frac{1}{P_{X|U}(x|Y)}\Big| X=x\Big] = \min_{y \in \mathcal{U}}\log\frac{1}{P_{X|U}(x|y)}
\end{align}
for every $x \in \mathcal{X}$. Define
\begin{align}
T(x):= \argmax_{y'}P_{X|U}(x|y').
\end{align}
Then, we have that 
\begin{align}
\mathbb{E}[d(X,Y)] &= D_{\min}\notag\\
&\iff
\forall x \in \mathcal{X}:\ \text{supp}\big(W(\cdot\mid x)\big)\subseteq T(x).
\end{align}
Thus, the RDF for the log-likelihood based problem $(X,P_{X|U})$ at $D_{\min}$ is
\begin{align}
R_{\ell\mkern-2mu\ell}(D_{\min})
&=
\min_{\substack{
W_{Y|X}:\\
\forall x, \text{ }\text{supp}(W_{Y|X}(\cdot|x))\subseteq T(x)
}}
I(X;Y)\label{eq:rdmin}
\end{align}
For a given $P_X$ and $W_{Y|X}$, the mutual information between $X$ and $Y$ can be written as
\begin{align}
    I(X;Y)= \min_{Q_Y} \mathbb{E}_X[D_{\text{KL}}(W_{Y|X}(\cdot|X)\|Q_Y)]
\end{align}
Thus,~\eqref{eq:rdmin} can be written as
\begin{align}
R_{\ell\mkern-2mu\ell}(D_{\min})&=
\min_{\substack{
W_{Y|X}:\\
\forall x, \text{ }\text{supp}(W_{Y|X}(\cdot|x))\subseteq T(x)
}}
\min_{Q_Y \in \Delta(\mathcal{U})}\ \mathbb{E}_X[D_{\text{KL}}(W_{Y|X}\|Q_Y)]\label{eq:rdmin2}\\
&=\min_{Q_Y \in \Delta(\mathcal{U})} \min_{\substack{
W_{Y|X}:\\
\forall x, \text{ }\text{supp}(W_{Y|X}(\cdot|x))\subseteq T(x)
}}
\mathbb{E}_X[D_{\text{KL}}(W_{Y|X}\|Q_Y)]\label{eq:rdmin3}\\
&= \min_{Q_Y \in \Delta(\mathcal{U})}\mathbb{E}_X \left[
\min_{\substack{
W_{Y|X}(\cdot|X):\\
\text{supp}(W_{Y|X}(\cdot|X))\subseteq T(X)
}}
D_{\mathrm{KL}}\left(W_{Y|X}(\cdot|X)\,\|\,Q_Y\right)
\right].\label{eq:rdminfinal}
\end{align}
Solving the optimization inside the expectation in \eqref{eq:rdminfinal} is equivalent to finding the I-projection of $Q_Y$ onto the simplex $\Delta(\mathcal{Y}|X=x)$ such that $\text{supp}(W_{Y|X}(\cdot|x))\subseteq T(x)$.

We will now use the method of Lagrange multipliers and apply KKT conditions to solve the optimization problem. We obtain that the optimal forward channel $W_{Y|X}(\cdot|x)$ is
\begin{align}
    W_{Y|X}(y|x)= \frac{Q_Y(y)}{Q_Y(T(x))}
\end{align}
for $y\in T(x)$ and $0$ otherwise.
Thus,~\eqref{eq:rdminfinal}  reduces to
\begin{align}
R_{\ell\mkern-2mu\ell}(D_{\min})    &= \min_{Q_Y \in \Delta(\mathcal{U})}\mathbb{E}_X \Bigg[
-\log \sum_{y \in T(X)}Q_Y(y)
\Bigg]\label{eq:rdminfinal2}
\end{align}
It is interesting to note here that, given $x \in \mathcal{X}$, if $T(x)$ is a singleton set, the optimal decoder at $D_{\min}$ is a maximum likelihood decoder. However, if $|T(x)|>1$, it randomizes among the symbols in $\mathcal{U}$ maximizing the likelihood under $P_{X|U}$.
This completes our proof for $R_{\ell\mkern-2mu\ell}(D_{\min})$.\\

\noindent\emph{Proof for $D_{\max}$ and $R(D_{\max})$:}
$D_{\max}$ is the smallest $D$ such that $\forall$ $D\geq D_{\max}$, we have $R(D)=0$. We will use the distortion-rate function to compute $D_{\max}$, i.e.,
\begin{align}
    D(0) = \min_{\substack{
W_{Y|X}:\\
I(X;Y)=0
}} \mathbb{E}_{XY}[-\log P_{X|U}(X|Y)].
\end{align}
We will first prove a lower bound on $D(0)$ and then provide a reconstruction of $Y$, such that the lower bound is achieved.

Since, $X$ and $Y$ are independent (due to the constraint $I(X;Y)=0$), we have
\begin{align}
    D(0)&= \min_{\substack{
W_{Y|X}:\\
X \perp Y
}} \sum_{x}P_X(x)\sum_{y}Q_Y(y)\log\frac{1}{P_{X|U}(x|y)}\\
&=\min_{Q_Y} \sum_{y}Q_Y(y)\sum_{x}P_X(x)\log\frac{1}{P_{X|U}(x|y)}\\
&=\min_{Q_Y} \sum_{y}Q_Y(y)\mathbb{E}_X\Big[-\log{P_{X|U}(X|y)}\Big]\\
& \geq \min_{y \in \mathcal{U}}\mathbb{E}_X\Big[-\log{P_{X|U}(X|y)}\Big]\label{eq:finaldmax}
\end{align}
For a decoder which outputs
\begin{align*}
    Y =
        u^{*} \in \argmin_{u \in \mathcal{U}}\mathbb{E}_X[-\log{P_{X|U}(X|u)}]&\text{ w.p. 1,}
\end{align*}
we observe that the expected distortion is always $\min_{y \in \mathcal{U}}\mathbb{E}_X[-\log{P_{X|U}(X|y)}]$, matching~\eqref{eq:finaldmax} with equality. Thus, we have that
\begin{align}
    D_{\max}=\min_{y \in \mathcal{U}}\mathbb{E}_X\Big[-\log{P_{X|U}(X|y)}\Big]
\end{align}
and $R(D_{\max})=0$.\\

\noindent\emph{Proof for $D^{*}$ and $R(D^{*})$:} Let $(U,X)\sim P_{U,X}$ be consistent with $(X,P_{X|U})$. At distortion $D^{*}=H(X|U)$, we have
\begin{align}
        R_{\ell\mkern-2mu\ell}(D)&= \min_{\substack{W_{Y|X}:\\H(X|Y)+\mathbb{E}_{Y}[D_\mathrm{KL}(W_{X|Y}(\cdot|Y)||P_{X|U}(\cdot|Y))] \leq H(X|U)}}I(X;Y) \\
        &=H(X)-\max_{\substack{W_{Y|X}:\\H(X|Y)+\mathbb{E}_{Y}[D_\mathrm{KL}(W_{X|Y}(\cdot|Y)||P_{X|U}(\cdot|Y))]\leq H(X|U)}}H(X|Y)\\
        & \geq H(X) - \max_{\substack{W_{Y|X}:\\H(X|Y)\leq H(X|U)}}H(X|Y)\\
        &=H(X)-H(X|U) = I(X;U)
\end{align}
Now, we will show that there exists a test channel (achievability analysis) such that the lower bound $I(X;U)$ is achieved with equality. Consider the test channel  $W_{Y|X}$ defined by
\begin{align}
W_{Y|X}(y|x)\triangleq P_{U|X}(y|x),\qquad y\in\mathcal{U}.
\end{align}
With this choice, the induced joint distribution of $(X,Y)\sim W_{X,Y}$ coincides with that of $(X,U)\sim P_{X,U}$, hence
\begin{align}
I(X;Y)=I(X;U)
\end{align}
Moreover, by consistency we have $P_{X|Y}(\cdot|y)=P_{X|U}(\cdot|y)$ for all $y$, which implies that
\begin{align}
\mathbb{E}_{Y}\left[D_{\mathrm{KL}}\big(W_{X|Y}(\cdot|Y)\,\big\|\,P_{X|U}(\cdot|Y)\big)\right]=0
\end{align}
Therefore, under the log-likelihood distortion,
\begin{align}
\mathbb{E}_{XY}[d_{\ell\mkern-2mu\ell}(X,Y)]
&=\mathbb{E}_{XY}\left[-\log P_{X|U}(X|Y)\right] \nonumber\\
&=H(X|Y)+\mathbb{E}_{Y}\left[D_{\mathrm{KL}}\big(W_{X|Y}(\cdot|Y)\,\big\|\,P_{X|U}(\cdot|Y)\big)\right] \nonumber\\
&=H(X|Y)=H(X|U)
\end{align}
Thus the distortion constraint is satisfied with equality at $D^*=H(X|U)$, and we obtain
\begin{align}
R_{\ell\mkern-2mu\ell}(D^*) \le I(X;Y)=I(X;U).
\end{align}
Combined with the (converse) lower bound $R_{\ell\mkern-2mu\ell}(D^*)\ge I(X;U)$, this yields
\begin{align}
R_{\ell\mkern-2mu\ell}(H(X|U))=I(X;U).
\end{align}

\subsection{Proof of Theorem~\ref{thm:logloss}}\label{app:logloss}
Given $(U,X) \sim P_{U,X}$, the RDF under log-likelihood loss for $(X,P_{X|U})$, $R_{\ell\mkern-2mu\ell}(D)$ for every $D \in [D_{\min}, D_{\max}]$ is given by
\begin{align}
      R_{\ell\mkern-2mu\ell}(D)=  \min_{W_{Y|X}:\mathbb{E}_{XY}[-\log P_{X|U}(x|y)] \leq D}I(X;Y)\label{eq:rsd}
\end{align}
Considering the distortion constraint, we have
\begin{align}
    D &\geq \mathbb{E}_{XY}[-\log P_{X|U}(X|Y)]\\
    &= \mathbb{E}_Y[E_{X|Y}[-\log P_{X|U}(X|Y)]] \\
    &=\mathbb{E}_Y\Bigg[E_{X|Y}\Bigg[\log \frac{1}{W_{X|Y}(X|Y)}+ \log \frac{W_{X|Y}(X|Y)}{P_{X|U}(X|Y)}\Bigg]\Bigg]\\
    &=H(X|Y)+ \mathbb{E}_{Y}[D_{\mathrm{KL}}(W_{X|Y}(\cdot|Y)\|P_{X|U}(\cdot|Y))].
\end{align}
Therefore, $R_{\ell\mkern-2mu\ell}(D)$ in~\eqref{eq:rsd} can be written as
\begin{align}
    R_{\ell\mkern-2mu\ell}(D)= \hspace{-2mm}\min_{\substack{W_{Y|X}:\\H(X|Y)+\mathbb{E}_{Y}[D_\mathrm{KL}(W_{X|Y}(\cdot|Y)||P_{X|U}(\cdot|Y))] \leq D}}\hspace{-2mm}I(X;Y). \label{eq:comp:srd2}
\end{align}
Recall that the optimization problem for the rate distortion function under log-loss is as follows
\begin{align}
    R_{\ell}(D)&=\min_{W_{Y|X}: H(X|Y) \leq D} I(X;Y)\label{eq:comp:llrd2}\\
    &=\begin{cases}
    H(X)-D \text{ }&;\text{ if  } \text{ } 0 \leq D \leq H(X)\\
    0 \text{ }&;\text{ if  } \text{ }  D > H(X)\\
     \end{cases}
\end{align}
On comparing the optimization constraint in~\eqref{eq:comp:srd2} and~\eqref{eq:comp:llrd2}, we observe that if the constraint is satisfied for RDF under log-likelihood loss for $(X,P_{X|U})$ , it is also satisfied for the log-loss RDF for $X$ i.e., the feasible set of $W_{Y|X}$ in~\eqref{eq:comp:srd2} is a subset of the feasible set of $W_{Y|X}$ in~\eqref{eq:comp:llrd2}. Since, the log-loss rate distortion function minimizes over a larger set, it can only make the optimum smaller. Thus, we have the following lower bound
\begin{align}
    R_{\ell\mkern-2mu\ell}(D) \geq R_{\ell}(D)
\end{align}
for every $D \in [D_{\min},D_{\max}]$. Note that we have $D_{\min}>0$ and $D_{\max} \geq H(X)$. For a consistent $(U,X)$, at distortion $D^{*}= H(X|U)$, we observe that $R_{\ell}(D^{*})=I(U;X)$. From Theorem~\ref{thm:prop1}, we also know that $R_{\ell\mkern-2mu\ell}(D^{*})=I(U;X)$ yielding the equality.

\subsection{Proof of Theorem~\ref{thm:dual}}\label{app:dual}

We present our proof for discrete alphabets $\mathcal{X}$ and $\mathcal{Y}$. However, it can be extended for continuous alphabets.
We will use the Lagrange multipliers and KKT conditions to solve for the optimization problem of the rate-distortion function $R(D)$ in it's dual form. Under the assumption that conditions in Theorem~\ref{thm:dual} hold and using strong duality, we prove our result.

For $(X,d)$ the rate-distortion function in the primal form is defined as
\begin{align}
R(D)=\min_{W_{Y|X}} \quad & I(X;Y) \\
\text{s.t.} \quad 
& \mathbb{E}[d(X,Y)] \le D, \notag \\
& \sum_{y} W_{Y|X}(y \mid x) = 1, \quad \forall x \in \mathcal{X}, \notag \\
& W_{Y|X}(y \mid x) \ge 0 \notag.
\end{align}
For $\lambda > 0$ and $\nu(x) \in \mathbb{R}$, define the Lagrangian
\begin{align}
    L(W,\lambda,\nu):= &I(X;Y)+ \lambda(\mathbb{E}(d(x,y))-D)\notag \\ & \hspace{5mm}+\sum_{x}\nu(x)\left(\sum_{y}W_{Y|X}(y|x)-1\right).
\end{align}
On using the stationarity of KKT conditions, and differentiating $L(W,\lambda,\nu)$, we have
\begin{align}
    \frac{\partial L(W,\lambda,\nu)}{\partial W}&= 0\\
    P_X(x)\log\frac{W^{*}_{Y|X}(y|x)}{Q_Y(y)}+\lambda P_X(x)d(x,y)+\nu(x)&=0\label{eq:stat}
\end{align}
Let $f(x):=e^{-\nu(x)/P_X(x)} \geq 0$. From the stationarity condition in \eqref{eq:stat}, the minimizing conditional distribution
$W^{*}_{Y|X}$ must be of the form
\begin{align}
W^{*}_{Y|X}(y|x)
= \frac{Q_Y(y)e^{-\lambda d(x,y)}}{f(x)}\label{eq:forward}
\end{align}
where $f(x)$ is an $x$-dependent normalization term introduced by the
Lagrange multipliers $\nu(x)$. Enforcing the primal feasibility constraint
$\sum_y W^{*}_{Y|X}(y|x)=1$ yields
\begin{align}
f(x)=f(x,\lambda)=\sum_y Q_Y(y)e^{-\lambda d(x,y)}\label{eq:fff}
\end{align}
Therefore, the corresponding backward channel induced by Bayes' rule is 
\begin{align}
W^{*}_{X|Y}(x|y)
&= \frac{P_X(x)W^{*}_{Y|X}(y|x)}{Q_Y(y)}\\
&= \frac{P_X(x)e^{-\lambda d(x,y)}}{f(x,\lambda)} \label{eq:backward}
\end{align}
Let $\tilde{\mu}(x,\lambda):=\frac{P_X(x)}{f(x,\lambda)}$. Now, for each $x \in \mathcal{X}$, $f(x,\lambda)=\sum_y Q_Y(y)e^{-\lambda d(x,y)}$ enforces that $W^{*}_{Y|X}$ is a valid conditional distribution i.e., $\sum_{y}W^{*}_{Y|X}(y|x)=1$. Similarly, for $W^{*}_{X|Y}$ to be a valid conditional distribution it must hold that $\sum_{x}\tilde{\mu}(\lambda,x)e^{-\lambda d(x,y)}=1$.

Thus, the KKT stationarity conditions determine the exponential-tilting form of the optimal forward and backward
channel. However, it does not ensure that there exists a reconstruction marginal $Q_Y$ such that
this expression is self-consistent (i.e., $Q_Y$ equals the marginal induced by $P_X$ and
$W^{\star}_{Y|X}$) and all KKT feasibility/normalization constraints are simultaneously satisfied.
Thus, we must additionally verify the existence of a $Q_Y$ for which the optimal forward
and backward channels form a valid coupling.

Now, from the assumption in Theorem~\ref{thm:dual} we know that for every $\lambda \in \mathcal{I}$, there exists a coupling $P^{\lambda}_{X,Y}$ with $X \sim P_X$ and 
\begin{align}
    P^{\lambda}_{X|Y}(x|y) = \mu(x,\lambda)e^{-\lambda d(x,y)}.
\end{align}
Since $ P^{\lambda}_{X|Y}$ is a valid conditional distribution, we have that for all $y \in \mathcal{Y}$
\begin{align}
    \sum_{x}\mu(x,\lambda)e^{-\lambda d(x,y)} =1.
\end{align}
Moreover, $X \sim P_X$ under the coupling $P^{\lambda}_{X,Y}$. Thus, 
\begin{align}
   P_X(x) &= \sum_{y}P^{\lambda}_Y(y)P^{\lambda}_{X|Y}(x|y)\\
   &=\sum_{y}P^{\lambda}_Y(y)\mu(x,\lambda)e^{-\lambda d(x,y)}.\label{eq:aasdf}
\end{align}
Let $Q_Y := P^{\lambda}_{Y}$. By condition (ii) of the theorem, for any $\lambda \in \mathcal{I}$, there exists a valid coupling, which implies that $Q_Y$ is a valid, non-negative probability distribution. Then,~\eqref{eq:aasdf} reduces to
\begin{align}
    \sum_{y}Q_Y(y)e^{-\lambda d(x,y)} = \frac{P_X(x)}{\mu(x,\lambda)}.
\end{align}
for all $x \in \mathcal{X}$.
Comparing with~\eqref{eq:fff}, we see that under this choice of $Q_Y$
\begin{align}
    f(x,\lambda):= \frac{P_X(x)}{\tilde{\mu}(x,\lambda)}=\frac{P_X(x)}{\mu(x,\lambda)}.
\end{align}
Therefore, we have $\tilde{\mu}(x,\lambda)=\mu(x,\lambda)$; the optimal backward channel $W^{*}_{X|Y}$ and some optimal forward channel $W^{*}_{Y|X}$ satisfy,
\begin{align}
    W^{*}_{Y|X}(y|x) &= P^{\lambda}_{Y|X}(y|x) =  \frac{Q_Y(y)\mu(x,\lambda)e^{-\lambda d(x,y)}}{P_X(x)}\label{eq:accd}\\
   W^{*}_{X|Y}(x|y)&=  P^{\lambda}_{X|Y}(x|y) = \mu(x,\lambda)e^{-\lambda d(x,y)}.
\end{align}

We will now show that the $R(D)$ can be expressed as an optimization problem over a single variable.
Let the dual function be
\begin{align}
    g(\lambda,\nu) = \min_{W_{Y|X}} L(W,\lambda,\nu).
\end{align}
Since strong duality holds, the RDF $R(D)$ can be found by maximizing the dual function $g(\lambda, \nu)$ over the set of feasible parameters $\lambda$. From the above, this set is $\mathcal{I}$, where the specific coupling structure exists. Thus,
\begin{align}\label{eq:dual}
    R(D)= \max_{\lambda \in \mathcal{I}, \nu}g(\lambda,\nu).
\end{align}

We know that the optimizing $W_{Y|X}$ is of the form in~\eqref{eq:accd}. Thus, at the minimizer $W^{*}_{Y|X} $ we have
\begin{align}
    g(\lambda,\nu)&=g(\lambda)\\&= I(X;Y)+\lambda(\mathbb{E}(d(x,y))-D)\\
    &=\sum_{x,y}P_X(x)W^{*}_{Y|X}(y|x)\log\frac{W^{*}_{Y|X}(y|x)}{Q_Y(y)}\notag \\& \hspace{10mm }+\lambda\sum_{x,y}P_X(x)W^{*}_{Y|X}(y|x)d(x,y)-\lambda D\\
    &=\sum_{x,y}P_X(x)W^{*}_{Y|X}(y|x)\log\frac{\mu(x,\lambda)e^{-\lambda d(x,y)}}{P_X(x)}\notag \\& \hspace{10mm }+\lambda\sum_{x,y}P_X(x)W^{*}_{Y|X}(y|x)d(x,y)-\lambda D\\
    &=\sum_{x}P_X(x)\log\frac{\mu(x,\lambda)}{P_X(x)}\left(\sum_{y}W^{*}_{Y|X}(y|x)\right)\notag \\&\hspace{55mm}-\lambda D\\
    &=\sum_{x}P_X(x)\log\frac{\mu(x,\lambda)}{P_X(x)} -\lambda D\\
    &= H(X) + \mathbb{E}_X[\log(\mu(X,\lambda))]-\lambda D .   
\end{align}

On substituting back to the dual expression for $R(D)$ in~\eqref{eq:dual}, we have that
\begin{align}
    R(D)=\max_{\lambda \in \mathcal{I}} \Big(H(X) + \mathbb{E}_X[\log(\mu(X,\lambda))] -\lambda D \Big).
\end{align}
This completes our proof.

\subsection{Proof of Theorem~\ref{thm:perc}}\label{app:perception}
\begin{proof}
Consider the rate-distortion-perception problem $(X,d,\gamma)$. On solving the optimization in~\eqref{eq:rdp-perfect1} for $R(D,0)$ using Lagrangian and applying the KKT conditions, we have that there exist $a,b: \mathcal{X}\to\mathbb{R}$, unique up to additive constants, such that the unique optimal coupling $W^{*}_{XY}$ is given by
    \begin{align}
    \label{eq:rdp-structure}
        \log W^{*}_{X,Y}(x,y) = -\lambda d(x,y) + a(x) + b(y).
    \end{align}
    Fix such $a,b$ and an arbitrary letter $x_0\in\mathcal{X}$. Define
    \begin{align}
        a^*(x) &= a(x) + \frac{b(x_0)-a(x_0)}{2}\\
        b^*(y) &= b(y) - \frac{b(x_0)-a(x_0)}{2}.
    \end{align}
    Thus, there exist unique $a^*, b^*$ satisfying \eqref{eq:rdp-structure} such that $a^*(x_0) = b^*(x_0)$.
    By the uniqueness of $W^{*}_{X,Y}(x,y)$ and the symmetry of $d(\cdot, \cdot)$ (note that the complete positivity (CP) of $V$ implies that it is symmetric, therefore $d(\cdot,\cdot)$ is also symmetric), we have that $W^{*}_{X,Y}(x,y) = W^{*}_{X,Y}(y,x)$ for all $x,y$.
    Then,
    \begin{align*}
        -\lambda d(x,y) + a^*(x) + b^*(y) = -\lambda d(y,x) + a^*(y) + b^*(x)
    \end{align*}
    Therefore, we then have $a^*(x) - b^*(x) = a^*(y)-b^*(y)$ for all $x,y$. Setting $x=x_0$, we see that $a^* = b^*$ point-wise. It implies that there exists a diagonal matrix $\Phi$ with strictly positive entries such that 
    \begin{align}
        W^{*}_{X,Y}(x,y) = [\Phi V(\cdot,\cdot)\Phi]_{x,y}
    \end{align}
    where multiplication on the RHS is matrix multiplication and exponentiation is performed \emph{element-wise}.
    Proposition~2.5 of \cite{berman_completely_2003} states that if a matrix $A$ is completely positive and $\Phi$ is positive diagonal, then $\Phi A\Phi$ is completely positive. We conclude that the joint distribution matrix associated with $W^{*}_{X,Y}$ is completely positive.

    Given two identical and jointly distributed random variables $(X,Y)~\sim P_{X,Y}$, there exists a $Z$ satisfying $X-Z-Y$ and $P_{X|Z}=P_{Y|Z}$ iff the coupling $P_{X,Y}$ is CP.
    Thus, this guarantees the existence of a $Z$ and the log-likelihood based rate-distortion framework $(Z,P_{Z|X})$. Such a $Z$ can be constructed via the CP-factorization of the optimal coupling $W^{*}_{X,Y}$ of $(X,d,\gamma)$. This  framework $(Z,P_{Z|X})$, at the special operating point, induces a reconstruction $Y^{*}$ such that the coupling $[W_{X,Z,Y^{*}}]_{XY^{*}}$ achieves the $R(D,0)$ for $(X,d,\gamma)$. 
\end{proof}
\subsection{Distortion measures satisfying Theorem~\ref{thm:perc}}\label{app:sufficient}
\subsubsection{Squared distance distortion measure}
    Let $d(x,y)$ be the squared distortion matrix i.e., $d(x,y)=(x-y)^2$. We will show that the induced element-wise exponential matrix $V:=\exp_H(-\lambda d(\cdot,\cdot))$ for this distortion measure is completely positive. From Schoenberg's theorem \cite{schoenberg}, it is known that any negative squared distance induces a matrix that is negative definite on the orthogonal complement of $\mathbf{1}$.
    Then, the matrix formed by $-\lambda d(x,y)$ is positive definite on the orthogonal complement of $\mathbf{1}$, for any $\lambda>0$.
    Therefore, we can write
    \begin{align}
        -\lambda d(\cdot, \cdot) = A + c \mathbf{1}^{\top} \mathbf{1}
    \end{align}
    for some positive definite matrix $A$ and (possibly negative) real $c$.
    Now,
    \begin{align}
       V= \exp_H(-\lambda d(\cdot, \cdot)) = \exp_H(A)\odot\exp_H(c\mathbf{1}^{\top} \mathbf{1})
    \end{align}
    where $\odot$ denotes the element-wise (or Hadamard) product of matrices.
    By Theorem~2.30 of \cite{berman_completely_2003}, we have that $\exp_H(A)$ is completely positive.
    We have the trivial factorization $\exp_H(c\mathbf{1}^{\top} \mathbf{1}) = (\exp_H(c/2)\mathbf{1})^{\top}\exp_H(c/2)\mathbf{1}$ which has element-wise nonnegative factors, so $\exp_H(c\mathbf{1}^\top \mathbf{1})$ is completely positive.
    By Corollary~2.2 of \cite{berman_completely_2003} (Hadamard product of CP matrices is also CP), we conclude that the exponential matrix $V$ is completely positive for every $\lambda>0$.\\
    
\subsubsection{Hamming distortion measure}

Let $d(x,y)$ be the $q$-ary hamming distortion (where $|\mathcal{X}|=|\mathcal{Y}|=q$) matrix i.e., $d(x,y)=\mathbf{1}(x \neq y)$. We will show that the induced element-wise exponential matrix $V$ for this distortion measure is completely positive, for any $q\geq 2$. Our proof is based on constructing a matrix $B$ satisfying conditions in Definition~\ref{def:cp} such that $V = BB^{\top}$. Define \(a:=e^{-\lambda}\in(0,1)\) such that
\begin{align}
V_{ij}=\exp_H(-\lambda d_{ij})=
\begin{cases}
1,& i=j,\\
a,& i\neq j.
\end{cases}
\end{align}
Equivalently, we have that
\begin{align}
V=(1-a)I_q+aJ_q,
\end{align}
where $I_q$ is the $q\times q$ identity matrix and $J_q=\mathbf{1}^\top\mathbf{1}$ is the all-ones matrix.
Define the matrix \(B\in\mathbb{R}_+^{q\times(q+1)}\) by

\begin{align}
B=\begin{bmatrix}
\sqrt{a} & \sqrt{1-a} & 0 & 0 & \cdots & 0\\
\sqrt{a} & 0 & \sqrt{1-a} & 0 & \cdots & 0\\
\sqrt{a} & 0 & 0 & \sqrt{1-a} & \cdots & 0\\
\vdots & \vdots & \vdots & \vdots & \ddots & \vdots\\
\sqrt{a} & 0 & 0 & 0 & \cdots & \sqrt{1-a}
\end{bmatrix}
\end{align}

i.e., the first column of \(B\) equals \(\sqrt{a}\,\mathbf{1}\) and all the remaining \(q\) columns form \(\sqrt{1-a}\,I_q\).
Let \(b_i\) denote the \(i\)-th row of \(B\). Then for each $i \leq q$ we have that
\begin{align}
(BB^\top)_{ii}=\langle b_i,b_i\rangle = a+(1-a)=1=V_{ii}.
\end{align}
where $\langle\cdot,\cdot\rangle$ denotes the inner product operator. For \(i\neq j\) the only common non-zero coordinate of \(b_i\) and \(b_j\) is the first element, hence
\begin{align}
(BB^\top)_{ij}=\langle b_i,b_j\rangle = (\sqrt{a})(\sqrt{a}) = a = V_{ij}.
\end{align}
Therefore, we have that \(BB^\top=V\) where every element of $B$ is non-negative. Thus, \(V\) is completely positive (CP) for every $\lambda>0$.

\end{document}